\documentclass[pdflatex,sn-basic,Numbered]{sn-jnl}

\usepackage{graphicx}%
\usepackage{multirow}%
\usepackage{amsmath,amssymb,amsfonts}
\usepackage{mathtools}%
\usepackage{amsthm}%
\usepackage{mathrsfs}%
\usepackage[title]{appendix}%
\usepackage{xcolor}%
\usepackage{textcomp}%
\usepackage{manyfoot}%
\usepackage{booktabs}%
\usepackage{array}%
\usepackage{algorithm}%
\usepackage{algorithmicx}%
\usepackage{algpseudocode}%
\usepackage{listings}%
\usepackage{tikz-cd}
\usetikzlibrary{arrows.meta}%
\usepackage[shortlabels]{enumitem}
\usepackage{stmaryrd}   

\newcommand{\Path}{\mathsf{Path}}
\newcommand{\Step}{\mathsf{Step}}
\newcommand{\StepTwo}{\mathsf{Step}_2}
\newcommand{\PathTwo}{\mathsf{Path}_2}
\newcommand{\SeqTwo}{\mathsf{Seq}_2}
\newcommand{\Nil}{\mathsf{Nil}}
\newcommand{\Seq}{\mathsf{Seq}}
\newcommand{\concatop}{\mathbin{\bullet}}
\newcommand{\invop}{^{-1}}
\newcommand{\ap}{\mathsf{ap}}
\newcommand{\refl}{\mathsf{refl}}
\newcommand{\sym}{\mathsf{sym}}
\newcommand{\trans}{\mathsf{trans}}
\newcommand{\Beta}{\mathsf{Beta}}
\newcommand{\BetaTwo}{\mathsf{Beta}_2}
\newcommand{\EtaTwo}{\mathsf{Eta}_2}
\newcommand{\LamCongTwo}{\mathsf{LamCong}_2}
\newcommand{\Eta}{\mathsf{Eta}}
\newcommand{\ApCong}{\mathsf{ApCong}}
\newcommand{\LamCong}{\mathsf{LamCong}}

\newcommand{\EtaExpand}{\mathsf{EtaExpand}}
\newcommand{\IdFun}{\mathsf{IdFun}}
\newcommand{\etaPath}{\eta}
\newcommand{\lamCongPath}{\mathsf{lamCong}}
\newcommand{\funPath}{\mathsf{funPath}}
\newcommand{\homotopy}{\sim}
\newcommand{\HomotopyStep}{\mathsf{HomotopyStep}}
\newcommand{\HomotopyStepI}{\mathsf{HomotopyStep}^{\mathrm{I}}}
\newcommand{\HomotopyPath}{\mathsf{HomotopyPath}}
\newcommand{\HomotopyPathI}{\mathsf{HomotopyPath}^{\mathrm{I}}}
\newcommand{\EvalHSI}{\mathsf{EvalHSI}}
\newcommand{\EvalHPI}{\mathsf{EvalHPI}}
\newcommand{\EHSRefl}{\mathsf{EHS}_{\mathsf{refl}}}
\newcommand{\EHSBeta}{\mathsf{EHS}_{\mathsf{Beta}}}
\newcommand{\EHSEta}{\mathsf{EHS}_{\mathsf{Eta}}}
\newcommand{\EHSApCong}{\mathsf{EHS}_{\mathsf{ApCong}}}
\newcommand{\EHSLamCong}{\mathsf{EHS}_{\mathsf{LamCong}}}
\newcommand{\EHSSym}{\mathsf{EHS}_{\mathsf{Sym}}}
\newcommand{\EHSTrans}{\mathsf{EHS}_{\mathsf{Trans}}}
\newcommand{\EHNilI}{\mathsf{EH}_{\mathsf{Nil}}}
\newcommand{\EHSeqI}{\mathsf{EH}_{\mathsf{Seq}}}
\newcommand{\whiskL}{\mathsf{whisk}_L}
\newcommand{\whiskR}{\mathsf{whisk}_R}
\newcommand{\assoc}{\mathsf{assoc}}
\newcommand{\leftInv}{\mathsf{leftInv}}
\newcommand{\rightInv}{\mathsf{rightInv}}
\newcommand{\leftId}{\mathsf{leftId}}
\newcommand{\rightId}{\mathsf{rightId}}
\newcommand{\invDist}{\mathsf{invDist}}
\newcommand{\invInv}{\mathsf{invInv}}
\newcommand{\cancelL}{\mathsf{cancel}}

\newcommand{\parity}{\mathsf{parity}}
\newcommand{\pathToStep}{\mathsf{pathToStep}}
\newcommand{\lembed}{\mathsf{lembed}}
\newcommand{\concatTwo}{\mathsf{concat}_2}
\newcommand{\invTwo}{\mathsf{inv}_2}

\newcommand{\NatStepI}{\mathsf{NatStepI}}
\newcommand{\NatPathI}{\mathsf{NatPathI}}

\newcommand{\defineq}{\equiv}

\theoremstyle{plain}
\newtheorem{theorem}{Theorem}[section]
\newtheorem{proposition}[theorem]{Proposition}
\newtheorem{lemma}[theorem]{Lemma}
\newtheorem{corollary}[theorem]{Corollary}
\theoremstyle{definition}
\newtheorem{definition}[theorem]{Definition}
\newtheorem{example}[theorem]{Example}
\theoremstyle{remark}
\newtheorem{remark}[theorem]{Remark}
\tikzcdset{every diagram/.append style={row sep=large, column sep=large},
 every label/.append style={font=\small}, every cell/.append style={font=\small}}
\AtBeginDocument{%
 \renewenvironment{proof}[1][\proofname]{\par\pushQED{\qed}%
 \normalfont\topsep6pt\trivlist\item[\hskip\labelsep\itshape #1.]\ignorespaces%
 }{\popQED\endtrivlist\ignorespacesafterend}%
}
\newcommand{\ev}{\mathsf{ev}}
\newcommand{\Nat}{\mathsf{Nat}}

\newcommand{\code}[1]{\texttt{#1}}
\newcommand{\Ftwo}{\mathbb{F}_2}
\hypersetup{pdftitle={A Theory of a Two-Dimensional Typed Lambda Calculus}, pdfauthor={Daniel O. Martinez-Rivillas; Arthur F. Ramos; Ruy J. G. B. de Queiroz}, pdfsubject={Typed conversion evidence and certified naturality}}

\begin{document}

\title[A Theory of a Two-Dimensional Typed Lambda Calculus]{A Theory of a Two-Dimensional Typed Lambda Calculus}

\author[1]{\fnm{Daniel O.} \sur{Martínez-Rivillas}}
\author[2]{\fnm{Arthur F.} \sur{Ramos}}
\author[3]{\fnm{Ruy J. G. B.} \sur{de Queiroz}}

\affil[1]{\orgdiv{Departamento de Matemáticas}, \orgname{Universidad Militar Nueva Granada (UMNG)}, \orgaddress{\country{Colombia}}}
\affil[2]{\orgname{Microsoft}, \orgaddress{\country{USA}}}
\affil[3]{\orgdiv{Centro de Informática (CIn)}, \orgname{Universidade Federal de Pernambuco (UFPE)}, \orgaddress{\city{Recife}, \state{Pernambuco}, \country{Brazil}}}

\abstract{We present a typed two-dimensional $\lambda$-calculus whose equality evidence is \emph{computational}: a path between two terms is an explicit finite sequence of one-step conversions (the $\beta$- and $\eta$-contractions, the congruences, and the structural rules), and every property of paths is proved \emph{by recursion over that sequence}, with any step as a base case --- in deliberate contrast with Martin-L\"of type theory, where identity is generated by reflexivity alone and all properties go through the non-computational $J$-eliminator. The higher structure is imported from the $2\beta$- and $2\eta$-conversions of the theory of an arbitrary higher $\lambda$-model: we obtain 2-dimensional coherence laws, computable naturality of homotopies (via inductive homotopies and their explicit evaluations), a 2-dimensional path type with transport, and a parity invariant that proves the system consistent and \emph{really intensional}: the $\beta$- and $\eta$-contractions are provably distinct evidence, while in the native syntax of Idris (core MLTT) they are identified by definitional equality. Commutative diagrams accompany the main constructions, and the theory is fully formalized in Idris 2. A parallel Lean formalization is published in the Palomar registry \cite{palomar2026lean}. A philosophical reading closes the paper: constructivism in the BHK sense, proof-relevant intensionality, and the boundary between syntax and semantics, drawn relative to MLTT and HoTT.}

\maketitle
\raggedbottom
\allowdisplaybreaks[0]
\clubpenalty=10000
\widowpenalty=10000

\section{Introduction}\label{sec:intro}

\subsection{Motivation: what is evidence of equality?}

Consider two programs $f$ and $g$ of the same type $A \to B$. When may we say that $f$ and $g$ are \emph{equal}, and what is the \emph{evidence} for such a statement? In the ordinary practice of computation the answer is operational: two programs are equal when there is a chain of elementary \emph{conversions} transforming one into the other, such as the $\beta$-contraction
\[
(\lambda x.\, t)\, a \;\text{converts to}\; t[a/x],
\]
or the $\eta$-contraction
\[
\lambda x.\, t\, x \;\text{converts to}\; t .
\]
The evidence of equality is then the chain itself: a finite sequence of steps, which can be inspected, enumerated and manipulated. This is the idea of \emph{computational paths} put forward by de Queiroz and his collaborators \cite{deQueiroz2016}, and it is the point of departure of this paper.

The situation is different in Martin-L\"of type theory (MLTT), the standard foundation of dependently typed programming. There, for terms $a, b : A$ one forms the \emph{identity type} $\mathsf{Id}_A(a,b)$, whose canonical (and, for the intensional theory, the \emph{only}) constructor is reflexivity $\mathsf{refl}_a : \mathsf{Id}_A(a,a)$; every property of identity proofs must be derived from the $J$-eliminator, and it is enough --- in fact necessary --- to prove it in the single case $e \equiv \mathsf{refl}_a$. Three features of this design concern us here.

\begin{enumerate}[(1)]
  \item \textbf{Only one base case.} The induction principle of the identity type has \emph{reflexivity as its only base step}. A computation like a $\beta$-contraction is not itself a constructor of $\mathsf{Id}_A(a,b)$; it is merely the \emph{translation} of $\mathsf{refl}$ along the definitional equality of the ambient core. Nothing in the syntax of MLTT distinguishes ``$f$ equals $g$ because both reduce to $h$'' from ``$f$ equals $g$ because $g$ is $f$''. The computational content of equality is erased.
  \item \textbf{Definitional equality conflates evidence.} In the native core of Idris (which is MLTT), the two terms $(\lambda x.\, u\, x)\, a$ and $u\, a$ are \emph{definitionally equal}. Thus the $\eta$-contraction of $u$ and the $\beta$-contraction of $u$ at $a$ are identified by the theory, even though they are witnessed by different rewriting steps.
  \item \textbf{Non-computational elimination.} The $J$-eliminator is a \emph{postulate} of the theory, not an algorithm. It gives no procedure that, from an identity witness, produces the transformations performed; semantic functions (arbitrary maps $A \to B$) escape it entirely.
\end{enumerate}

This paper develops an alternative: a typed two-dimensional $\lambda$-calculus in which the evidence of equality is an explicit object, a \emph{path} built from one-step conversions, and in which every property of paths is proved \emph{by recursion over the sequence of steps}, with \emph{any} step --- $\beta$, $\eta$, a congruence, or a structural rule --- available as a base case. The higher (2-dimensional) structure is imported from the $2\beta$- and $2\eta$-conversions of the theory of an arbitrary higher $\lambda$-model \cite{martinez2023arbitrary}, whose equality theory subsumes that of homotopic $\lambda$-models and of extensional Kan complexes \cite{martinez2023groupoid,martinez2022domain}.

\subsection{Relationship to higher lambda models}\label{sec:related}

Martínez-Rivillas and de Queiroz's theory of an arbitrary higher $\lambda$-model introduces higher conversion evidence and studies its interpretation in extensional Kan complexes \cite{martinez2023arbitrary}. The beta and eta naturality squares on pp.~48--49 motivate our typed generators. Their Example~2.14 already contrasts beta- and eta-based routes. 

The $K_\infty$ homotopy $\lambda$-model supplies a concrete semantic construction and studies higher beta/eta conversions, including the comparison in Example~4.16 \cite{martinez2026kinfty}. Recursive completion in higher models develops a further syntactic account and a Lean~4 formalization \cite{martinez2026recursive}. Its Remark~8.4 restricts one separation analysis to a subsystem with forward generators, reflexivity and composition; its Theorem~8.7 concerns a canonical equality tower.

Our contribution relative to these three works is the combination of (i)~a typed  version of the $2\beta\eta$-contractions of \cite{martinez2023arbitrary}, (ii)~a recursive evaluator for an inductive homotopy language that constructs computational naturality certificates by means of the $2\beta\eta$-contractions, (iii)~direct and conjugation-based naturality witnesses, (iv)~a parity invariant checked against every constructor of the present $\StepTwo$, including inverse, congruence and whiskering rules, the parity invariant proves consistency and hence the intensionality theorem: the $\beta$- and $\eta$-contractions are distinct evidence in our system, while they are definitionally identified in MLTT, and (v)~ a philosophical reading: the theory as mathematical constructivism --- BHK-style evidence, proof-relevant intensionality, and the constructive boundary between syntax and semantics --- positioned with respect to MLTT and HoTT. 

\subsection{Scope and method}\label{sec:philosophy}

The theory has three levels of evidence: $\Step$ and $\Path$ record
conversions between MLTT/Idris host terms, $\StepTwo$ and $\PathTwo$ record
cells between parallel paths in our two-dimensional theory, and native MLTT
equality is used only for metatheoretic results such as the Boolean invariant.
Thus definitionally equal endpoints may still carry distinct path data. A
pointwise homotopy is semantic evidence, whereas an inductive presentation
exposes a derivation for recursion; the certified evaluators supply
naturality only for the latter.

Here ``two-dimensional'' concerns terms, paths and cells, not the order of
function types. The formalization uses MLTT/Idris host functions rather than
a separate syntax with binding and substitution, so it proves neither
normalization nor completeness for a standalone calculus nor all bicategorical
or infinity-groupoid equations. Sections~\ref{sec:prelim}--\ref{sec:algebra}
give the presentation and path algebra, Sections~\ref{sec:homotopy}--\ref{sec:path2}
the certified naturality and transport constructions, Section~\ref{sec:consistency}
the non-collapse result, and Section~\ref{sec:philosophical} the philosophical
interpretation; Appendix~\ref{sec:appendix} relates them to the code.

\section{Typed conversion evidence}\label{sec:prelim}

We work over MLTT/Idris host types $A,B,\ldots$ and MLTT/Idris host terms $x:A$, with ordinary function types $A\to B$. All displayed rules are schematic in these types. The implementation quantifies over Idris \code{Type}; it does not restrict endpoints to a freely generated simply typed term language. The lambda-calculus notation recalls the usual beta and eta conversions \cite{barendregt1984}, but the evidence of these conversions is represented separately from the endpoint expressions.

Write $\defineq$ for MLTT/Idris host definitional equality, $\lembed s$ for a singleton path, and $p\concatop q$ for traversal of $p$ followed by $q$. When more than two paths are composed we show parentheses where the chosen association matters. A square in a diagram abbreviates an inhabitant of a specified $\StepTwo$ or $\PathTwo$; it does not assert literal equality of the boundary paths.

\subsection{Judgemental conversion and labelled evidence}\label{sec:judgemental}

For comparison, take the presentation of MLTT with judgemental beta and function eta in \cite[Secs.~1.2--1.3]{hottbook}. These are rules establishing judgements such as $\Gamma\vdash M\defineq N:B$, rather than constructors of distinct identity terms labelled beta and eta. Identity introduction and type conversion then permit
\[
\frac{\Gamma\vdash M\defineq N:B}
     {\Gamma\vdash\mathsf{refl}^{\mathrm{Id}}_M:\mathrm{Id}_B(M,N)}.
\]
This witness records no choice of derivation of the premise. Thus two conversion derivations may have the same judgement as conclusion and supply the same reflexivity witness; this is not an equation between derivation objects internal to MLTT. The qualification on eta matters: we compare with a presentation in which function eta is judgemental, without presuming that convention in every variant of type theory.

The $\Beta$ and $\Eta$ constructors instead make conversion justifications into indexed data. Definitional equality of their endpoints does not erase their tags. The distinction becomes mathematically substantive when it survives the declared relations between paths, as proved in Corollary~\ref{cor:contra2}. The implementation retains the MLTT/Idris host's own conversion judgement; it adds a separate language of evidence rather than replacing that judgement.

\subsection{Steps}\label{sec:steps}

\begin{definition}[Conversion steps]\label{def:step}
The indexed family $\Step(x,y)$ has seven constructors:
\begin{align*}
\Beta(f,a)&:\Step((\lambda x.\,f\,x)\,a,f\,a),
 &\Eta(f)&:\Step(\lambda x.\,f\,x,f),\\
\ApCong(f,s)&:\Step(fx,fy) &&(s:\Step(x,y)),\\
\LamCong(h)&:\Step(\lambda x.\,u\,x,\lambda x.\,v\,x)
 &&(h:(x:A)\to\Step(ux,vx)),\\
\refl_x&:\Step(x,x),\\
\sym(s)&:\Step(y,x) &&(s:\Step(x,y)),\\
\trans(s,t)&:\Step(x,z) &&(s:\Step(x,y),\ t:\Step(y,z)).
\end{align*}
\end{definition}

Unlike MLTT, where propositional equality is freely generated by reflexivity (and eliminated by $J$), the steps above are \emph{computational}: Beta and Eta are actual contractions of terms, ApCong and LamCong are congruences, and Refl, Sym, Trans are the meta-structural closure. The step relation is therefore a syntax-directed rewriting relation on terms, not an inductively generated equality. Two consequences deserve emphasis, as they will shape the whole paper.

\begin{enumerate}[(a)]
	\item \textbf{Every step carries its own meaning.} A $\beta$-step is not an $\eta$-step, and neither is a reflexivity step; the type $\Step(x,y)$ keeps them apart. In MLTT all identity witnesses of $\mathsf{Id}_A(a,a)$ are forced to be $\mathsf{refl}$, so this information cannot even be expressed.
	\item \textbf{The rules are reversible but not symmetric.} Symmetry is added \emph{explicitly} as the structural constructor $\sym$: the step from $f\, x$ back to $(\lambda y.\, f\, y)\, x$ is $\sym(\Beta(f,x))$, a genuinely different step from $\Beta(f,x)$ itself. Section~\ref{sec:consistency} exploits precisely this asymmetry, via the parity invariant, to prove that $\beta$- and $\eta$-evidence can never be identified.
\end{enumerate}

\subsection{Paths and their operations}\label{sec:paths}

\begin{definition}[Paths]\label{def:path}
For $x,y:A$, a path is a finite outer sequence of steps:
\[
\Nil_x:\Path(x,x),\qquad
\Seq(s,p):\Path(x,z)\quad(s:\Step(x,y),\ p:\Path(y,z)).
\]

A path is thus a finite sequence of steps. Pictorially, a path $p : \Path(x,y)$ is a directed chain
\begin{center}
	\begin{tikzcd}[row sep=small, column sep=small]
		x = x_0 \arrow[r, "s_1"] & x_1 \arrow[r, "s_2"] & x_2 \arrow[r] & \cdots \arrow[r, "s_n"] & x_n = y
	\end{tikzcd}
\end{center}
\noindent where each label $s_i$ is a step of the type $\Step$ --- a $\beta$, an $\eta$, a congruence, or a structural rule --- and the endpoint terms $x_i$ are part of the datum; the empty sequence $\Nil_x$ is the path of length zero. Notice the methodological inversion with respect to MLTT: there, one proves a property of $e : \mathsf{Id}_A(x,y)$ by reducing $e$ to $\mathsf{refl}_x$; here, one unfolds the chain $s_1, \dots, s_n$, and \emph{every step type} contributes a case --- a property may be checked at any step, not only at the trivial one.
\end{definition}

\begin{definition}[Operations]\label{def:pathops}
Concatenation, inversion, congruence and compression are given by structural recursion:
\begin{align*}
\Nil\concatop q&=q,&\Seq(s,p)\concatop q&=\Seq(s,p\concatop q),\\
\lembed s&=\Seq(s,\Nil),\\
\Nil\invop&=\Nil,&\Seq(s,p)\invop&=p\invop\concatop\lembed{\sym(s)},\\
\ap_f(\Nil)&=\Nil,&\ap_f(\Seq(s,p))&=\Seq(\ApCong(f,s),\ap_f(p)),\\
\pathToStep(\Nil)&=\refl,&\pathToStep(\Seq(s,p))&=\trans(s,\pathToStep(p)).
\end{align*}
They have the expected endpoint types: inversion reverses endpoints, $\ap_f$ maps them through $f$, and $\pathToStep:\Path(x,y)\to\Step(x,y)$.
\end{definition}

Induction on $\Path$ has two cases, $\Nil$ and $\Seq(s,p)$. A further induction on $s$ may expose any of the seven step constructors.

\begin{definition}[Function paths]\label{def:etaexpand}\label{def:funpath}
Let $\EtaExpand(t)=\lambda x.\,t\,x$, $\IdFun(t)=t$, and $\etaPath(f)=\lembed{\Eta(f)}$. For a pointwise family $h:(x:A)\to\Path(fx,gx)$, put
\begin{align*}
\lamCongPath(h)&=\lembed{\LamCong(\lambda x.\,\pathToStep(hx))},\\
\funPath(f,g,h)&=\etaPath(f)\invop\concatop
                   (\lamCongPath(h)\concatop\etaPath(g)):\Path(f,g).
\end{align*}
\end{definition}

This gives a path of functions in our two-dimensional relation. No law identifying application of this path with the original family $h$ is asserted. Also, $\EtaExpand$ is just a definition in the MLTT/Idris host: expanding it in $\ap_{\EtaExpand}(p)$ produces no new beta or eta evidence. Explicit constructors, not alternative spellings of endpoint functions, carry that evidence.

\section{Two-dimensional structure: Step2}\label{sec:step2}

\begin{definition}[Step2: 2-cells between paths]\label{def:step2}
For $p, q : \Path(x,y)$, the type $\StepTwo(p,q)$ of \emph{2-cells} from $p$ to $q$ is generated by the following rules.

\smallskip
\noindent\textbf{Reduction squares (typed 2-beta and 2-eta, after \cite{martinez2023arbitrary}).} The beta and eta rules are typed analogues motivated by the model-theoretic squares of \cite[Cor.~2.18]{martinez2023arbitrary}; the remaining rules specify the structural closure chosen for this presentation. Each one is a \emph{square}: two parallel paths (top-then-right, and left-then-bottom) with a 2-cell filling the square. We draw each square as a commutative diagram; the label in the middle is the name of the 2-cell.
\begin{enumerate}[(i)]
  \item \textbf{Reflexivity commutation} $\mathsf{CRefl}^{\mathsf{step}}_2(s)$, for $s : \Step(x,y)$:
  \[ \mathsf{CRefl}^{\mathsf{step}}_2(s) : \StepTwo(\lembed{\refl} \concatop \lembed(s),\ \lembed(s) \concatop \lembed{\refl}). \]
  \item \textbf{2-Beta} $\BetaTwo(s)$, for $u : A \to B$, $s : \Step(a,b)$:
  \begin{multline*}
   \BetaTwo(s) : \StepTwo\bigl(\ap_{\lambda x. u x}(\lembed(s)) \concatop \lembed(\Beta(u,b)),\\
   \lembed(\Beta(u,a)) \concatop \ap_u(\lembed(s))\bigr).
  \end{multline*}
  \begin{center}
  \begin{tikzcd}[column sep=large, row sep=large]
  (\lambda x.\, u\, x)\, a
    \arrow[r, "{\ap_{\lambda x. u x}(\lembed s)}"]
    \arrow[d, "{\lembed(\Beta(u,a))}"']
    \arrow[dr, phantom, "{\BetaTwo(s)}" {description, font=\footnotesize}, color=blue]
  & (\lambda x.\, u\, x)\, b \arrow[d, "{\lembed(\Beta(u,b))}"] \\
  u\, a \arrow[r, "{\ap_u(\lembed s)}"']
  & u\, b
  \end{tikzcd}
  \end{center}
  Reading the square top-right: the $\beta$-contraction of $u$ \emph{after} pushing $s$ through the $\eta$-expanded function; bottom-left: pushing $s$ through $u$ \emph{after} the contraction. $\BetaTwo(s)$ is our typed analogue of the higher $2\beta$ square: it commutes $\beta$ with the congruence of $u$. In the Idris source this generator is \code{Beta2Step}; its displayed orientation is already the common naturality orientation.
  \item \textbf{2-Eta (step generator)} $\EtaTwo(s)$, for $s : \Step(u,v)$:
  \begin{multline*}
  \EtaTwo(s) : \StepTwo\bigl(\ap_{\EtaExpand}(\lembed(s))\concatop\lembed(\Eta(v)),\\
                         \lembed(\Eta(u))\concatop\lembed(s)\bigr),
  \end{multline*}
  The path-level square for an arbitrary $p$ is derived below by induction; it
  is not an additional primitive constructor of $\StepTwo$.
  \begin{center}
  \begin{tikzcd}[column sep=large, row sep=large]
  \lambda x.\, u\, x
    \arrow[r, "{\ap_{\lambda t. \lambda x. t x}(\lembed s)}"]
    \arrow[d, "{\lembed(\Eta(u))}"']
    \arrow[dr, phantom, "{\EtaTwo(s)}" {description, font=\footnotesize}, color=blue]
  & \lambda x.\, v\, x \arrow[d, "{\lembed(\Eta(v))}"] \\
  u \arrow[r, "{\lembed s}"']
  & v
  \end{tikzcd}
  \end{center}
  Here the top and bottom arrows live \emph{in the type $A \to B$}: the top edge applies the eta-expansion congruence to the step $s : \Step(u,v)$, the bottom edge is $s$ itself, and the vertical edges are the $\eta$-contractions; the 2-cell commutes $\eta$ with the congruence. In the Idris source this generator is \code{Eta2Step}, with the same naturality orientation; the path-level proof \code{eta2} invokes it directly.
  \item \textbf{Naturality of Lambda congruence}. For a presented homotopy $h^I:\HomotopyStepI(u,v)$ (Definition~\ref{def:homotopy}), let $h=\llbracket h^I\rrbracket:(x:A)\to\Step(u\,x,v\,x)$ be its structural evaluation. For $s : \Step(a,b)$, put $L_t=\lembed(\ApCong(\lambda z.\,zt,\LamCong(h)))$. Then
  \begin{multline*}
   \LamCongTwo(h^I,s) :
   \StepTwo\bigl(\ap_{\lambda x.\,u x}(\lembed s)\concatop L_b,\,
   L_a\concatop\ap_{\lambda x.\,v x}(\lembed s)\bigr).
  \end{multline*}
  \begin{center}
  \begin{tikzcd}[column sep=large, row sep=large]
  (\lambda x.\,u\,x)\,a
    \arrow[r, "{\ap_{\lambda x.\,u x}(\lembed s)}"]
    \arrow[d, "{L_a}"']
    \arrow[dr, phantom, "{\LamCongTwo(h^I,s)}" {description, font=\footnotesize}, color=blue]
  & (\lambda x.\,u\,x)\,b \arrow[d, "{L_b}"] \\
  (\lambda x.\,v\,x)\,a \arrow[r, "{\ap_{\lambda x.\,v x}(\lembed s)}"']
  & (\lambda x.\,v\,x)\,b
  \end{tikzcd}
  \end{center}
  The vertical edges evaluate the lambda-congruence $\LamCong(\llbracket h^I\rrbracket)$ at $a$ and $b$. Requiring $h^I$ makes the global inductive derivation part of this generator; its direction uses the displayed eta-expanded functions on both boundaries.
\end{enumerate}

\smallskip
\noindent\textbf{Structural cell generators.}
\begin{enumerate}[(i)]
  \item \textbf{Reflexivity} $\mathsf{Refl}_2 : \StepTwo(p,p)$; \textbf{symmetry} $\sym_2 : \StepTwo(p,q) \to \StepTwo(q,p)$; \textbf{transitivity} $\trans_2 : \StepTwo(p,q) \to \StepTwo(q,r) \to \StepTwo(p,r)$.
  \item \textbf{Structural coherence} (a meta-step represents the corresponding composite):
  \begin{align*}
    \mathsf{ReflStep}_2 &: \StepTwo(\lembed(\refl_x), \Nil), \\
    \sym\mathsf{Step}_2(s) &: \StepTwo(\lembed(\sym(s)), \lembed(s)\invop), \\
    \trans\mathsf{Step}_2(s,t) &: \StepTwo(\lembed(\trans(s,t)), \lembed(s) \concatop \lembed(t)).
  \end{align*}
  \item \textbf{Elementary cancellation}
  \begin{align*}
    \leftInv_2(s) &: \StepTwo(\lembed(s)\invop \concatop \lembed(s),\ \Nil), \\
    \rightInv_2(s) &: \StepTwo(\lembed(s) \concatop \lembed(s)\invop,\ \Nil), \\
    \sym\sym_2(s) &: \StepTwo(\lembed(\sym(\sym(s))),\ \lembed(s)).
  \end{align*}
  \item \textbf{Functoriality of congruence}, with $\alpha:\StepTwo(p,q)$:
  \begin{align*}
    \ap\mathsf{CongIdStep}_{2}(s) &: \StepTwo(\ap_{\lambda t. t}(\lembed(s)),\ \lembed(s)),\\
    \ap\mathsf{CongStep}_2(u,\alpha) &: \StepTwo(\ap_u(p),\ap_u(q)).
  \end{align*}
  \begin{multline*}
    \ap\mathsf{CongComposeStep}_{2}(u,f,s):
    \StepTwo\bigl(\ap_u(\ap_f(\lembed(s))),\\
                      \ap_{\lambda z. u(fz)}(\lembed(s))\bigr).
  \end{multline*}
  \item \textbf{Elementary left and right extensions}
  \[
  \mathsf{Whisk}_{L}\mathsf{Step}_{2}(\alpha) : \StepTwo(\Seq(s,p), \Seq(s,q)) \quad\text{for } \alpha : \StepTwo(p,q),
  \]
  \[
  \mathsf{Whisk}_{R}\mathsf{Step}_{2}(\alpha, s) :
  \StepTwo(p \concatop \lembed(s), q \concatop \lembed(s)).
  \]
\end{enumerate}
\end{definition}

\begin{remark}[Induction over rules and sequences]\label{rem:induction}
The proofs distinguish induction on a derivation from induction on a sequence of derivations. For $\Step$, beta, eta and reflexivity have no recursive $\Step$ premises; congruence, symmetry and transitivity do. Their induction cases therefore receive hypotheses for the constituent evidence, pointwise in the lambda-congruence case. Similarly, induction on $\StepTwo$ checks each declared square and treats recursive cell constructors such as $\sym_2$, $\trans_2$ and whiskering using their induction hypotheses.

Sequence induction has the schematic clauses
\begin{align*}
F(\Nil)&=d,&F(\Seq(s,p))&=c(s,p,F(p)),\\
G(\Nil_2)&=d_2,&G(\SeqTwo(\alpha,R))&=c_2(\alpha,R,G(R)).
\end{align*}
Here endpoint indices are implicit and $d,c,d_2,c_2$ have the corresponding dependent types. The sequence clause can use a previously established result for every step or cell; it need not inspect that entry again. Thus ``checking the basic rules and extending to sequences'' describes a layered proof strategy, rather than making every rule a non-recursive base constructor. The parity argument in Section~\ref{sec:consistency} exemplifies this strategy. In particular, the arbitrary-path 2-eta square is derived by the coherence construction in Theorem~\ref{thm:eta2}; only $\EtaTwo$ is a primitive reduction square.
\end{remark}

\begin{remark}[Design choice]\label{rem:whisk}
These rules generate witnesses, rather than equations between all witnesses. The elementary left and right extensions, $\mathsf{Whisk}_{L}\mathsf{Step}_{2}$ and $\mathsf{Whisk}_{R}\mathsf{Step}_{2}$, are primitive. Whiskering by arbitrary paths is derived recursively; together they yield horizontal composition without making that composition primitive (cf.\ Proposition~\ref{prop:trans2horizontal}).
\end{remark}

\section{The algebra of paths: coherence laws}\label{sec:algebra}

The following constructions use structural recursion and the declared cell generators. Similar groupoid operations on computational paths have established antecedents \cite{deveras2025groupoid}; the purpose here is to identify exactly which witnesses this presentation supports. Equations marked $\defineq$ arise from evaluation of MLTT/Idris host definitions. Other transformations are witnessed by cells and cannot silently be used as definitional equalities.

\begin{lemma}[Cancellation of adjacent steps]\label{lem:cancel}
For $s : \Step(x,y)$ and $q : \Path(y,z)$:
\[
\cancelL(s,q) : \StepTwo\bigl(\lembed(s)\invop \concatop \Seq(s,q),\ q\bigr).
\]
\end{lemma}

\begin{proof}
Right-whisker the elementary left-inverse cell by $q$:
\[
\whiskR(\leftInv_2(s),q):
\StepTwo((\lembed s\invop\concatop\lembed s)\concatop q,\Nil\concatop q).
\]
Because both singleton paths have explicit constructors, these boundaries reduce to $\lembed s\invop\concatop\Seq(s,q)$ and $q$. The derived operation $\whiskR$ performs the required recursion on $q$, so this lemma needs no separate induction.
\end{proof}

\begin{lemma}[Distributivity of concatenation]\label{lem:distrib}
For $s : \Step(x,y)$, $p : \Path(y,z)$, $q : \Path(z,w)$:
\[
\mathsf{distrib}(s,p,q) : \StepTwo((\Seq(s,p)) \concatop q,\ \Seq(s, p \concatop q)),
\]
which holds definitionally ($\mathsf{Refl}_2$), since concatenation is defined by recursion on its first argument.
\end{lemma}

\begin{theorem}[Total 2-dimensional inverses]\label{thm:inv2}
For every path $p : \Path(x,y)$:
\[
\leftInv_2(p) : \StepTwo(p\invop \concatop p,\ \Nil), \qquad
\rightInv_2(p) : \StepTwo(p \concatop p\invop,\ \Nil).
\]
\end{theorem}

\begin{proof}
Both by recursion on $p$. For $p = \Nil$ both 2-cells are $\mathsf{Refl}_2$. For $p = \Seq(s,p')$ we have, by the defining recursions of $\invop$ and $\concatop$,
\[
p\invop \concatop p \;=\; \bigl({p'}\invop \concatop \lembed(s)\invop\bigr) \concatop \Seq(s,p'),
\]
and the pasting
\begin{multline*}
\bigl({p'}\invop \concatop \lembed(s)\invop\bigr) \concatop \Seq(s,p')
\;\xrightarrow{\ \assoc_2\ }\;
{p'}\invop \concatop \bigl(\lembed(s)\invop \concatop \Seq(s,p')\bigr) \\
\;\xrightarrow{\ \whiskL({p'}\invop,\ \cancelL(s,p'))\ }\;
{p'}\invop \concatop p'
\;\xrightarrow{\ \leftInv_2(p')\ }\;\Nil .
\end{multline*}
For the right inverse, first reassociate
\[
\Seq(s,p')\concatop\Seq(s,p')\invop
\quad\text{to}\quad
\lembed s\concatop((p'\concatop {p'}\invop)\concatop\lembed s\invop).
\]
This uses a symmetric associativity cell inside the head frame, not merely unfolding. Apply the recursive right-inverse cell to $p'$, framed by the two singleton paths, and close with the elementary right-inverse cell for $s$.
\end{proof}

\begin{theorem}[2-dimensional identities]\label{thm:id2}
For every path $p : \Path(x,y)$:
\[
\leftId_2(p) : \StepTwo(\Nil \concatop p,\ p), \qquad
\rightId_2(p) : \StepTwo(p \concatop \Nil,\ p).
\]
\end{theorem}

\begin{proof}
$\leftId_2(p) = \mathsf{Refl}_2$: concatenation recurses on its \emph{first} argument, so $\Nil \concatop p = p$ definitionally. For the right identity the recursion is non-trivial: $\rightId_2(\Nil) = \mathsf{Refl}_2$, and $\rightId_2(\Seq(s,p')) = \mathsf{Whisk}_{L}\mathsf{Step}_{2}(\rightId_2(p'))$, since $(\Seq(s,p')) \concatop \Nil = \Seq(s, p' \concatop \Nil)$ by definition.
\end{proof}

\begin{theorem}[Associativity]\label{thm:assoc}
For paths $p : \Path(x,y)$, $q : \Path(y,z)$, $r : \Path(z,w)$:
\[
\assoc_2(p,q,r) : \StepTwo((p \concatop q) \concatop r,\ p \concatop (q \concatop r)).
\]
\end{theorem}

\begin{proof}
$\assoc_2(\Nil,q,r) = \mathsf{Refl}_2$, and $\assoc_2(\Seq(s,p'),q,r) = \mathsf{Whisk}_{L}\mathsf{Step}_{2}(\assoc_2(p',q,r))$: both sides are $\Seq(s, -)$ applied to the two regroupings of $p'$, $q$, $r$. This yields an explicit cell by structural recursion; it does not assert judgmental associativity for an arbitrary variable path.
\end{proof}

\begin{definition}[Whiskerings]\label{def:whisk}
\begin{enumerate}[(i)]
  \item $\whiskL(p, \alpha) : \StepTwo(p \concatop q,\ p \concatop r)$ for $\alpha : \StepTwo(q,r)$, defined by recursion on $p$;
  \item $\whiskR(\alpha, r) : \StepTwo(p \concatop r,\ q \concatop r)$ for $\alpha : \StepTwo(p,q)$, defined by recursion on $r$. The empty case uses the right-identity cells around $\alpha$; the sequence case reassociates, applies $\mathsf{Whisk}_{R}\mathsf{Step}_{2}$ to its head, recurses on its tail, and reassociates the result.
\end{enumerate}
\end{definition}

\begin{proposition}[Horizontal composition]\label{prop:trans2horizontal}
For 2-cells $\alpha_1 : \StepTwo(p_1,q_1)$ and $\alpha_2 : \StepTwo(p_2,q_2)$ with $p_1,q_1 : \Path(x,y)$ and $p_2,q_2 : \Path(y,z)$:
\[
\trans_2\bigl(\whiskR(\alpha_1, p_2),\ \whiskL(q_1, \alpha_2)\bigr) : \StepTwo(p_1 \concatop p_2,\ q_1 \concatop q_2).
\]
\end{proposition}

\begin{proof}
Frame $\alpha_1$ on the right by $p_2$, frame $\alpha_2$ on the left by $q_1$, and sequence the two results: the horizontal pasting
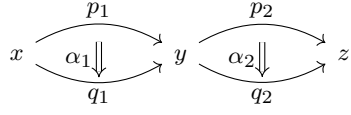
\begin{center}
\begin{tikzcd}[column sep=huge]
x \arrow[r, "p_1", ""{name=L, below}, bend left=30] \arrow[r, "q_1"', ""{name=M, below}, bend right=30] &
y \arrow[r, "p_2", ""{name=N, below}, bend left=30] \arrow[r, "q_2"', ""{name=O, below}, bend right=30] &
z
\arrow[from=L, to=M, "{\alpha_1}"', shorten <=2pt, shorten >=2pt, Rightarrow]
\arrow[from=N, to=O, "{\alpha_2}"', shorten <=2pt, shorten >=2pt, Rightarrow]
\end{tikzcd}
\end{center}
\noindent read left to right. The derived right whiskering and the recursively defined left whiskering therefore yield horizontal composition without adding it as a generator.
\end{proof}

\begin{proposition}[Reflexivity commutation for paths]\label{prop:crefl2}
For every path $q : \Path(x,y)$:
\[
\mathsf{CRefl}_2(q) : \StepTwo(\lembed{\refl} \concatop q,\ q \concatop \lembed{\refl}).
\]
\end{proposition}

\begin{proof}
$\mathsf{CRefl}_2(\Nil) = \mathsf{Refl}_2$, and for $q = \Seq(s,q')$ one pastes the elementary commutation square $\mathsf{CRefl}^{\mathsf{step}}_2(s)$ (framed on the right by $q'$) with the recursive call $\mathsf{CRefl}_2(q')$ (framed on the left by $\lembed(s)$).
\end{proof}

\begin{theorem}[2-Beta for paths]\label{thm:beta2}
For $u : A \to B$ and every path $p : \Path(a,b)$:
\[
\Beta_2(u,p) : \StepTwo\bigl(\ap_{\lambda x. u x}(p) \concatop \lembed(\Beta(u,b)),\ \lembed(\Beta(u,a)) \concatop \ap_u(p)\bigr).
\]
\end{theorem}

\begin{proof}
$\Beta_2(u,\Nil) = \mathsf{Refl}_2$. For $p = \Seq(s,p')$ the two paths are, up to regrouping by $\assoc_2$,
\begin{multline*}
\ap_{\lambda x. u x}(\lembed s) \concatop \ap_{\lambda x. u x}(p') \concatop \lembed(\Beta(u,b)),\\
\lembed(\Beta(u,a)) \concatop \ap_u(\lembed s) \concatop \ap_u(p').
\end{multline*}
One pastes the direct generator \code{Beta2Step} (written $\BetaTwo(s)$ above) in the middle and the recursive cell $\Beta_2(u,p')$ to its right, regrouping by $\assoc_2$ at both ends. For a path with $n$ entries, this recursion pastes $n$ elementary beta squares, with the beta component at each intermediate endpoint shared by neighbouring squares.
\end{proof}

\begin{theorem}[2-Eta for paths]\label{thm:eta2}
For every path $p : \Path(u,v)$:
\[
\Eta_2(p) : \StepTwo\bigl(\ap_{\EtaExpand}(p) \concatop \lembed(\Eta(v)),\ \lembed(\Eta(u)) \concatop p\bigr).
\]
\end{theorem}

\begin{proof}
Put $s_p=\pathToStep(p)$ and let
$c_p:\StepTwo(\lembed{s_p},p)$ be the certificate
\code{pathToStepCoherence(p)}.  This certificate is defined by structural
induction on $p$: the $\Nil$ case is $\mathsf{ReflStep}_2$, and the
$\Seq(s,p')$ case pastes $\mathsf{TransStep}_2$ with the inductive
coherence using left whiskering.  Apply $\ap$ to $c_p$, right-whisker by
$\lembed(\Eta(v))$, and reverse that cell.  Then paste the primitive square
$\EtaTwo(s_p)$ and left-whisker $c_p$ by $\lembed(\Eta(u))$:
\begin{align*}
\StepTwo(&\ap_{\EtaExpand}(p)\concatop\lembed(\Eta(v)),
          \ap_{\EtaExpand}(\lembed{s_p})\concatop\lembed(\Eta(v)))\\
\StepTwo(&\ap_{\EtaExpand}(\lembed{s_p})\concatop\lembed(\Eta(v)),
          \lembed(\Eta(u))\concatop\lembed{s_p})\\
\StepTwo(&\lembed(\Eta(u))\concatop\lembed{s_p},
          \lembed(\Eta(u))\concatop p).
\end{align*}
Transitivity gives the displayed result.  Thus the proof is constructive
induction on $p$; the path-level witness is assembled from the step
generator and structural coherence, rather than postulated as a new rule.
\end{proof}

\begin{theorem}[Naturality of Lambda congruence]\label{thm:lamcong2}
For $h^I:\HomotopyStepI(u,v)$, put $h=\llbracket h^I\rrbracket$. For every path $p : \Path(a,b)$:
\begin{multline*}
\LamCong_2(h^I,p) : \StepTwo\bigl(\ap_{\lambda x.\,u x}(p) \concatop \lembed(\ApCong(\lambda z. z\, b, \LamCong(h))),\\
\lembed(\ApCong(\lambda z. z\, a, \LamCong(h))) \concatop \ap_{\lambda x.\,v x}(p)\bigr).
\end{multline*}
\end{theorem}

\begin{proof}[Proof sketch]
 The elementary generator $\LamCongTwo(h^I,s)$ has the displayed orientation. The implementation assembles the strip by induction on $p$, using that generator and the recursive hypothesis in the same direction. Thus \code{lamCong2} requires the global presentation $h^I$ and exposes the displayed direction directly.
\end{proof}

\begin{theorem}[Functoriality of $\ap$]\label{thm:apcong}
\begin{enumerate}[(i)]
  \item $\ap\mathsf{CongId}(p) : \StepTwo(\ap_{\lambda t. t}(p),\ p)$ for every path $p$;
  \item $\ap\mathsf{CongCompose}(u, f, p) : \StepTwo(\ap_u(\ap_f(p)),\ \ap_{\lambda z. u(f\,z)}(p))$ for every path $p$.
  \item $\mathsf{apcongConcat}(f,p,q):\StepTwo(\ap_f(p\concatop q),\ap_f(p)\concatop\ap_f(q))$; also $\ap_f(\Nil)\defineq\Nil$.
\end{enumerate}
\end{theorem}

\begin{proof}[Proof sketch]
For (i): $p = \Nil$ gives $\mathsf{Refl}_2$; for $p = \Seq(s,p')$, paste $\mathsf{Whisk}_{L}\mathsf{Step}_{2}(\ap\mathsf{CongId}(p'))$ with the elementary square $\ap\mathsf{CongIdStep}_{2}(s)$ framed by $p'$. For (ii) the same pattern with $\ap\mathsf{CongComposeStep}_{2}(u,f,s)$. For (iii), induction on $p$ uses $\mathsf{Refl}_2$ and $\mathsf{Whisk}_{L}\mathsf{Step}_{2}$. The first two laws concern identity and composition of MLTT/Idris host functions; the third concerns composition of paths. Together they specify the functorial behaviour needed below.
\end{proof}

\begin{theorem}[Inverse distributes]\label{thm:invdist}
For paths $p : \Path(x,y)$, $q : \Path(y,z)$:
\[
\invDist(p,q) : \StepTwo\bigl((p \concatop q)\invop,\ q\invop \concatop p\invop\bigr).
\]
\end{theorem}

\begin{proof}[Proof sketch]
Recursion on $p$: for $\Nil$ use $\rightId_2$ (symmetrically); for $p = \Seq(s,p')$, unfold $(p \concatop q)\invop = (p' \concatop q)\invop \concatop \lembed(\sym s)$ by the defining recursion of $\invop$, paste the recursive cell $\invDist(p',q)$ framed by $\lembed(\sym s)$, and regroup with $\assoc_2$.
\end{proof}

\begin{theorem}[Involution]\label{thm:invinv}
For every path $p : \Path(x,y)$:
\[
\invInv(p) : \StepTwo\bigl((p\invop)\invop,\ p\bigr).
\]
\end{theorem}

\begin{proof}[Proof sketch]
For $\Nil$ use $\mathsf{Refl}_2$. If $p=\Seq(s,p')$, unfold $p\invop={p'}\invop\concatop\lembed{\sym(s)}$ and apply Theorem~\ref{thm:invdist} to obtain a cell to
\[
\lembed{\sym(s)}\invop\concatop({p'}\invop)\invop.
\]
The double-symmetry generator converts the first factor to $\lembed s$; frame the recursive involution cell by that factor. The result is $\Seq(s,p')$. The initial inverse-distribution step is a cell, not a definitional equality.
\end{proof}

\section{Certified homotopies and naturality}\label{sec:homotopy}

For functions $f,g:A\to B$, a pointwise path family is an inhabitant of
\[
(f\homotopy g):=(x:A)\to\Path(fx,gx).
\]
We call this a homotopy in our two-dimensional path relation. The term does not assert that a topological interpretation has been constructed. Its naturality at $p:\Path(x,y)$ would be a cell
\begin{equation}\label{eq:naturality}
\Nat(h,p):\StepTwo(\ap_f(p)\concatop hy,\ hx\concatop\ap_g(p)).
\end{equation}
For identity-valued homotopies, the corresponding statement is a standard consequence of path induction \cite[Lemma~2.4.3]{hottbook}. Our paths have different generators and elimination rules, so that argument cannot simply be reused without checking their cases.

\begin{definition}[Pointwise and presented homotopies]\label{def:homotopy}
A step family is $\HomotopyStep(f,g):=(x:A)\to\Step(fx,gx)$. A \emph{path in function space} $\HomotopyPath(f,g)$ is generated by $\mathsf{HNil}$ and $\mathsf{HSeq}(h,hp)$, where $h:\HomotopyStep(f,g)$ and $hp:\HomotopyPath(g,k)$ give a path from $f$ to $k$. Its pointwise evaluation is
\[
|\mathsf{HNil}|x=\Nil,\qquad
|\mathsf{HSeq}(h,hp)|x=\Seq(hx,|hp|x).
\]
The bars abbreviate \code{homotopyPathToHomotopy}.

The inductive presentation $\HomotopyStepI(f,g)$ has constructors
\begin{align*}
\mathsf{HSRefl}&:\HomotopyStepI(f,f),\\
\mathsf{HSBeta}(u)&:\HomotopyStepI(\lambda x.\,ux,u),\\
\mathsf{HSEta}&:\HomotopyStepI(\EtaExpand,\IdFun),\\
\mathsf{HSApCong}(v,e)&:\HomotopyStepI(v\circ f,v\circ g),\\
\mathsf{HSSym}(e)&:\HomotopyStepI(g,f),\\
\mathsf{HSTrans}(e,d)&:\HomotopyStepI(f,k),\\
\mathsf{HSLamCong}(e)&:\HomotopyStepI(\lambda x.\,fx,\lambda x.\,gx).
\end{align*}
Here $e:\HomotopyStepI(f,g)$, $d:\HomotopyStepI(g,k)$, and $\mathsf{HSEta}$ has domain and codomain the function type under consideration. The analogous sequence type $\HomotopyPathI$ has constructors $\mathsf{HNilI}$ and $\mathsf{HSeqI}$ using these presented step homotopies.
Structural recursion defines \code{evalHomotopyStepIHom}$(e)=\llbracket e\rrbracket:\HomotopyStep(f,g)$; each clause evaluates the corresponding presented constructor. In particular, \code{LamCong2Step} and \code{lamCong2} receive $e$ itself and use $\llbracket e\rrbracket$ only to form their boundary steps.
\end{definition}

The phrase ``path in function space'' refers only to the endpoints
$f\to g$ and the outer constructors $\mathsf{HNil}$ and $\mathsf{HSeq}$;
it does not prescribe how a square is drawn. In the naturality diagrams,
$h x$ and $h y$ are vertical edges, while $\ap_f(p)$ and $\ap_g(p)$ are the
top and bottom edges. We reserve \emph{horizontal} for horizontal
composition of 2-cells (Proposition~\ref{prop:trans2horizontal}).

Both function-space path types are inductive, but their entries expose
different data: $\HomotopyPath$ stores arbitrary pointwise step families,
whereas $\HomotopyPathI$ stores derivations from a specified grammar.
Thus $\HomotopyStep$ is semantic data without an inductive derivation for
recursive naturality, while $\HomotopyStepI$ exposes that derivation and its
certified evaluator yields \code{NatStepI}; for sequences,
$\HomotopyPathI$ with $\EvalHPI$ yields \code{NatPathI}.

The global-versus-pointwise distinction is the same at the sequence level:
$\HomotopyPath$ has a global outer sequence but arbitrary semantic entries,
whereas $\HomotopyPathI$ requires one global decomposition into presented
step homotopies. The term
\code{HSeq (\textbackslash x => pathToStep (h x)) HNil} can compress a pointwise family into
an ordinary $\HomotopyPath$, but it is not exported as a named operation and
does not normalize the family into elementary $\HomotopyPathI$ constructors.
Likewise, $\LamCong$ is a syntactic $\Step$ constructor whose premise may
contain an arbitrary homotopy, but its two-dimensional generator
\code{LamCong2Step} requires a $\HomotopyStepI$ presentation;
\code{funPath} does not supply such a normalization.

Consequently, naturality for arbitrary $\HomotopyStep$ and pointwise
homotopies is a semantic claim too strong to obtain by computation. In the
present development \code{StepNat} is constructively refuted; the certified
inductive routes above are the available naturality theorems. The file
\code{PathNonComputable.idr} records \code{NatPath} as the corresponding type
declaration:
\[
\begin{aligned}
\code{NatPath} &: (h:\homotopy(f,g))\to(p:\Path(x,y))\\
&\to \StepTwo(\,\ap_f(p)\concatop h y,\ h x\concatop\ap_g(p)\,).
\end{aligned}
\]
Here $h:(x:A)\to\Path(fx,gx)$ is arbitrary; the declaration has no defining
equation or inhabitant and is not a certified naturality theorem. The
refutation module formalizes the obstruction: \code{natPathToStepNat}
specializes such an operation to one-step homotopies and singleton paths, and
\code{notNatPath} derives \code{Void} from \code{notStepNat}.

\begin{definition}[Certified evaluation]\label{def:evalhsi}
For $e:\HomotopyStepI(f,g)$, the record $\EvalHSI(e)$ contains a family $h:\HomotopyStep(f,g)$ and, for every $p:\Path(x,y)$, a cell
\begin{equation}\label{eq:certnat}
\StepTwo(\ap_f(p)\concatop\lembed{hy},\ \lembed{hx}\concatop\ap_g(p)).
\end{equation}
Its fields are \code{evalHSIHom} and \code{evalHSINat}. The index $e$ does not itself constrain the first field by an evaluation equation: the record is a certificate interface. The recursive construction below selects a particular evaluation of each expression; uniqueness of arbitrary records is not claimed.
Here a \emph{certificate} is an Idris term inhabiting this dependent record:
it packages the evaluated pointwise family together with the corresponding
\code{Step2} witness of naturality.  For function-space presentations, the
analogous \code{EvalHPI} package contains the evaluated function-space path and
its \code{Path2} witness.  We call the resulting property
\emph{computationally certified naturality} when such a package is produced
by total structural recursion on the inductive presentation and checked by
the Idris type checker.  The certificate is therefore explicit data together
with its proof; it is not an assumption about an arbitrary pointwise family.
\end{definition}

\begin{theorem}[Canonical certified evaluation]\label{prop:evalhsi}
Every $e:\HomotopyStepI(f,g)$ has a recursively constructed certificate $\ev(e):\EvalHSI(e)$. Its pointwise family, written $\llbracket e\rrbracket x$, obeys
\begin{align*}
\llbracket\mathsf{HSRefl}\rrbracket x&=\refl,\\
\llbracket\mathsf{HSBeta}(u)\rrbracket x&=\Beta(u,x),\\
\llbracket\mathsf{HSEta}\rrbracket t&=\Eta(t),\\
\llbracket\mathsf{HSApCong}(v,e)\rrbracket x&=\ApCong(v,\llbracket e\rrbracket x),\\
\llbracket\mathsf{HSSym}(e)\rrbracket x&=\sym(\llbracket e\rrbracket x),\\
\llbracket\mathsf{HSTrans}(e,d)\rrbracket x&=\trans(\llbracket e\rrbracket x,\llbracket d\rrbracket x),\\
\llbracket\mathsf{HSLamCong}(e)\rrbracket x
 &=\ApCong(\lambda z.\,zx,\LamCong(\llbracket e\rrbracket)).
\end{align*}
In each case the certificate contains a proof of~\eqref{eq:certnat}.
\end{theorem}

\begin{proof}
Induct on $e$. The reflexivity case uses the symmetry of $\mathsf{CRefl}_2(\ap_f(p))$. The beta case uses $\Beta_2(u,p)$. For eta, first replace $\ap_{\IdFun}(p)$ by $p$ using congruence functoriality and then orient the resulting eta square as in~\eqref{eq:certnat}. The application case applies $\ap_v$ to the recursive square and uses the composition and concatenation laws of $\ap$. The lambda-congruence case is precisely $\LamCong_2(e,p)$: its boundaries contain $\LamCong(\llbracket e\rrbracket)$, so the required global presentation is explicit.

For symmetry, invert the two vertical steps of the recursive square, paste it in the reverse direction, and cancel the inserted inverse pairs. For transitivity, replace each $\trans$ step by its two-step path using $\trans\mathsf{Step}_2$, compose the two recursive squares, and compress the opposite boundary back to a $\trans$ step. These constructions are the total functions $\EHSRefl$, $\EHSBeta$, $\EHSEta$, $\EHSApCong$, $\EHSLamCong$, $\EHSSym$, and $\EHSTrans$. Their recursion is assembled by \code{evalHomotopyStepI}.
\end{proof}

\begin{corollary}[Naturality of evaluated step homotopies]\label{thm:nattransi}
For every presented step homotopy $e$ and every path $p$, the family $\llbracket e\rrbracket$ has the square
\begin{center}
\begin{tikzcd}
fx\arrow[r,"{\ap_f(p)}"]\arrow[d,"{\lembed{\llbracket e\rrbracket x}}"'] & fy\arrow[d,"{\lembed{\llbracket e\rrbracket y}}"]\\
gx\arrow[r,"{\ap_g(p)}"'] & gy.
\end{tikzcd}
\end{center}
\end{corollary}
\begin{proof}
Apply \code{evalHSINat} to the certificate constructed in Theorem~\ref{prop:evalhsi}. The more general wrapper \code{NatStepI} accepts any certificate and projects this field. Thus the wrapper does not assume that all pointwise families are natural; the theorem constructs the required certificate for the specified grammar.
\end{proof}

\begin{definition}[Evaluation of function-space presentations]\label{def:evalhpi}
The relation $\EvalHPI(e,hp)$ is generated by
\[
\EHNilI:\EvalHPI(\mathsf{HNilI},\mathsf{HNil})
\]
and, given $c:\EvalHSI(e_0)$ and $d:\EvalHPI(e_1,hp)$, by
\[
\EHSeqI(c,d):\EvalHPI(\mathsf{HSeqI}(e_0,e_1),\mathsf{HSeq}(\mathsf{evalHSIHom}(c),hp)).
\]
Applying the canonical step evaluator recursively gives, for every $e:\HomotopyPathI(f,g)$, a function-space path $hp$ together with $\EvalHPI(e,hp)$ (\code{evalHomotopyPathI}).
\end{definition}

The evaluated family $|hp|$ can now be studied without any global naturality premise. The next section records the resulting pastings as sequences of 2-cells.

\section{Sequences of 2-cells and transport}\label{sec:path2}

\begin{definition}[2-paths]\label{def:path2}\label{def:path2ops}
For parallel $p,q$, the family $\PathTwo(p,q)$ is generated by an empty sequence $\Nil_2:\PathTwo(p,p)$ and
\[
\SeqTwo(\alpha,R):\PathTwo(p,r)
\quad(\alpha:\StepTwo(p,q),\ R:\PathTwo(q,r)).
\]
Embedding of a cell, concatenation $\concatTwo$, reversal $\invTwo$, and left and right whiskering are defined by recursion on this sequence. They use, respectively, $\SeqTwo$, $\sym_2$, and the whiskerings of Section~\ref{sec:algebra}.
\end{definition}

Since $\StepTwo$ already has $\trans_2$, finite pastings could also be composed into one $\StepTwo$ witness. $\PathTwo$ is useful for retaining a sequence of such cells and recursing on that sequence; it is not necessary for the mere existence of a finite composite. No equality theory of 2-paths, or family of 3-cells comparing their presentations, is introduced.

\begin{theorem}[Direct naturality for function-space evaluations]\label{thm:nattranshp2}
Given $e:\HomotopyPathI(f,g)$, $hp:\HomotopyPath(f,g)$ and $c:\EvalHPI(e,hp)$, for every $p:\Path(x,y)$ there is a constructed witness
\[
\PathTwo(\ap_f(p)\concatop|hp|y,\ |hp|x\concatop\ap_g(p)).
\]
\end{theorem}
\begin{proof}
Induct on the evaluation evidence to construct the converse strip, then reverse that finite sequence with $\invTwo$. Thus the exported \code{NatPathI} has the displayed orientation. The construction uses only certificates produced by the grammar, without a premise for arbitrary pointwise families.
\end{proof}

\subsection{Conjugating paths}

\begin{definition}[Transport]\label{def:transp}
Fix endpoints $x,y:A$. For $hp:\HomotopyPath(f,g)$ define
\[
T_{hp}:\Path(fx,fy)\to\Path(gx,gy)
\]
by
\begin{align*}
T_{\mathsf{HNil}}(q)&=q,\\
T_{\mathsf{HSeq}(h,hp')}(q)
 &=T_{hp'}(\lembed{hx}\invop\concatop(q\concatop\lembed{hy})).
\end{align*}
The implementation \code{Transp2viaHP} also receives $p:\Path(x,y)$, which supplies the endpoints; the recursion uses $q$ and the components of $hp$. This is transport by conjugation in our two-dimensional path algebra, not the dependent eliminator for an MLTT/Idris host identity type.
\end{definition}

The operation follows the three other sides of a square:
\begin{center}
\begin{tikzcd}
fx\arrow[r,"q"]\arrow[d,"{\lembed{hx}}"']&fy\arrow[d,"{\lembed{hy}}"]\\
gx\arrow[r,"{\lembed{hx}\invop\concatop(q\concatop\lembed{hy})}"']&gy.
\end{tikzcd}
\end{center}
The bottom edge is a definition. A naturality or comparison cell is an additional result.

\begin{proposition}[Transport respects 2-paths]\label{prop:transpresp}
Every $R:\PathTwo(q,r)$ induces $\PathTwo(T_{hp}(q),T_{hp}(r))$.
\end{proposition}
\begin{proof}
Induct on $hp$. The empty case returns $R$. In a sequence case, right-whisker $R$ by $\lembed{hy}$, left-whisker it by $\lembed{hx}\invop$, and apply the recursive transport to the tail. This is \code{Transp2viaHPResp}; it needs no naturality certificate.
\end{proof}

\begin{lemma}[Conjugating an evaluated step]\label{lem:conjap}
If $c:\EvalHSI(e)$ has family $h$, then for every $p:\Path(x,y)$ there is a cell
\[
\StepTwo(\lembed{hx}\invop\concatop(\ap_f(p)\concatop\lembed{hy}),\ \ap_g(p)).
\]
\end{lemma}
\begin{proof}
Use the naturality cell of $c$ under the left frame $\lembed{hx}\invop$. It transforms the source into
$\lembed{hx}\invop\concatop(\lembed{hx}\concatop\ap_g(p))$.
Reassociate, apply the left-inverse cell, and remove the empty prefix. This is \code{ConjApcongStep2Eval}.
\end{proof}

\begin{theorem}[Transport of congruence]\label{thm:transpapcong}
For $c:\EvalHPI(e,hp)$ and $p:\Path(x,y)$ there is a witness
\[
\PathTwo(T_{hp}(\ap_f(p)),\ \ap_g(p)).
\]
\end{theorem}
\begin{proof}
Induct on $c$. The empty case is $\Nil_2$. For a sequence, Lemma~\ref{lem:conjap} reduces the first conjugation to the congruence for the intermediate function. Embed that cell and transport it through the remaining function-space path using Proposition~\ref{prop:transpresp}; concatenate the resulting 2-path with the recursive witness. This is \code{Transp2viaHPApcongEval}.
\end{proof}

\begin{theorem}[Recovering a square from transport]\label{thm:untransp}
Let $hp:\HomotopyPath(f,g)$, $q:\Path(fx,fy)$ and $r:\Path(gx,gy)$. From $R:\PathTwo(T_{hp}(q),r)$ one constructs
\[
\PathTwo(q\concatop|hp|y,\ |hp|x\concatop r).
\]
\end{theorem}
\begin{proof}
For $\mathsf{HNil}$, first remove the empty suffix of $q$ and then apply $R$. For $\mathsf{HSeq}(h,hp')$, apply the induction hypothesis to the conjugated path and the tail. Prepend $\lembed{hx}$ to that square, cancel the adjacent pair $\lembed{hx}\concatop\lembed{hx}\invop$ using the right-inverse cell, and reassociate the remaining prefixes and suffixes. The resulting boundaries are exactly those displayed. The implementation is \code{Untransp2viaHP}.
\end{proof}

\begin{corollary}[Naturality via transport]\label{thm:nattransvia}
For $c:\EvalHPI(e,hp)$, transport of $\ap_f(p)$ followed by Theorem~\ref{thm:untransp} constructs
\[
\PathTwo(\ap_f(p)\concatop|hp|y,\ |hp|x\concatop\ap_g(p)).
\]
\end{corollary}
\begin{proof}
Take $q=\ap_f(p)$, $r=\ap_g(p)$ and the witness of Theorem~\ref{thm:transpapcong}. Theorem~\ref{thm:untransp} has exactly the displayed boundary order. This is \code{NatTransHSviaTranspEval}.
\end{proof}

The direct, function-space, and transport constructions use the same boundary orientation: the composite through $f$ before $h_y$ is related to the composite through $g$ after $h_x$. The function-space route uses $\invTwo$ only to present its witness in that common direction; no equality between the resulting 2-path presentations is claimed.

\section{Consistency and intensionality: beta and eta evidence are distinct}\label{sec:consistency}

This section proves the negative results promised in the introduction: the theory is \emph{consistent} (the 2-cells do not collapse everything), and in particular the $\beta$- and $\eta$-contractions are \emph{provably distinct evidence} --- while in the native syntax of Idris, based on MLTT, the same two contractions are identified by definitional equality. The rule set admits many cells, but does not relate every pair of parallel paths. A structural grading into $\Ftwo=\{0,1\}$ proves this without invoking a topological model. We write $\oplus$ for addition modulo two (Boolean XOR).

\begin{definition}[Structural parity]\label{def:parity}
Define the step grading $\parity$ by
\begin{align*}
\parity(\Beta(f,x))&=1,&\parity(\Eta(f))&=0,\\
\parity(\ApCong(f,s))&=\parity(s),&\parity(\LamCong(h))&=0,\\
\parity(\refl)&=0,&\parity(\sym(s))&=\parity(s),\\
\parity(\trans(s,t))&=\parity(s)\oplus\parity(t).
\end{align*}
For paths, set $\parity(\Nil)=0$ and $\parity(\Seq(s,p))=\parity(s)\oplus\parity(p)$.
\end{definition}

This grading counts beta tags along the recursive branches exposed by the definition. In particular, it assigns zero to a lambda-congruence step without inspecting its function-valued premise. It is not the number of all beta reductions occurring in every term or hidden inside every witness. That choice is what makes the lambda-congruence squares compatible with the invariant.

\begin{lemma}[Compositionality]\label{lem:paritycomp}
For well-typed paths and functions,
\begin{align*}
\parity(p\concatop q)&=\parity(p)\oplus\parity(q),&
\parity(p\invop)&=\parity(p),\\
\parity(\ap_f(p))&=\parity(p),&
\parity(\lembed s)&=\parity(s).
\end{align*}
\end{lemma}
\begin{proof}
Structural induction using the defining equations and the unit, associativity and commutativity laws of XOR. For inversion, the recursive case is the sum of the tail's parity and that of the symmetric head; commutativity restores the original order. The corresponding functions are \code{pathParityConcat}, \code{pathParityInv}, \code{pathParityApcong}, and \code{pathParityLembed}.
\end{proof}

\begin{theorem}[Preservation by all 2-cells]\label{prop:parityinv}
If $\alpha:\StepTwo(p,q)$, then $\parity(p)=\parity(q)$. The same conclusion holds for $R:\PathTwo(p,q)$.
\end{theorem}
\begin{proof}
Induct on $\alpha$. The generating squares have the following boundary gradings, with $b=\parity(s)$ or $\parity(p)$ as appropriate:
\begin{center}
\begin{tabular}{@{}lll@{}}
\toprule
Generator & First boundary & Second boundary\\
\midrule
Reflexivity commutation & $0\oplus b$ & $b\oplus0$\\
Beta square & $b\oplus1$ & $1\oplus b$\\
Eta square (both variants) & $b\oplus0$ & $0\oplus b$\\
Lambda-congruence square & $0\oplus b$ & $b\oplus0$\\
Inverse cancellation & $b\oplus b$ & $0$\\
\bottomrule
\end{tabular}
\end{center}
Reflexive, symmetric and transitive cells use the corresponding rules of Boolean equality. The structural step-to-path cells respect the defining XOR equations. Double symmetry preserves $b$; identity and composite congruences preserve it by Lemma~\ref{lem:paritycomp}. Applying a function to a cell preserves the two boundary gradings. A fixed prefix or one-step suffix adds the same grading to each side, covering $\mathsf{Whisk}_{L}\mathsf{Step}_{2}$ and $\mathsf{Whisk}_{R}\mathsf{Step}_{2}$; arbitrary whiskering follows from the recursive constructions. These cases exhaust $\StepTwo$ and constitute \code{step2PreservesParity}.

For $\PathTwo$, induction on the cell sequence composes these equalities; the empty case is reflexivity. This is \code{path2PreservesParity}.
\end{proof}

\begin{corollary}[A family of unfillable squares]\label{cor:contra1}
For $u:A\to B$ and $p:\Path(a,b)$ there is no cell between
\begin{align*}
P&=\Seq(\ApCong(\lambda t.\,ta,\Eta(u)),\ap_u(p)),\\
Q&=\ap_{\lambda x.\,ux}(p)\concatop\lembed{\Beta(u,b)}.
\end{align*}
\end{corollary}
\begin{proof}
The gradings are $\parity(p)$ and $\parity(p)\oplus1$, respectively. These cannot be equal in $\Ftwo$ (\code{contraexample1}).
\end{proof}

\begin{corollary}[The key intensionality failure]\label{cor:contra1nil}\label{cor:contra2}
For $u:A\to B$ and $a:A$, put $M=(\lambda x.\,ux)\,a$. The certificates
\begin{align*}
b&=\Beta(u,a):\Step(M,ua),\\
e&=\ApCong(\lambda t.\,ta,\Eta(u)):\Step(M,ua)
\end{align*}
have the same written source and target. Their singleton paths $P_\beta=\lembed b$ and $P_\eta=\lembed e$ have neither a $\StepTwo$ nor a $\PathTwo$ witness between them. They are therefore distinct under the identification relation of our two-dimensional theory; \code{commonSourceNoNativePathEquality} also proves $P_\eta\ne P_\beta$ in MLTT's native Idris equality. This inequality is asserted only for this selected pair.
\end{corollary}
\begin{proof}
Beta contracts the application in $M$; eta contracts its function subterm in the context $[-]\,a$. The latter source is represented by $(\lambda t.\,ta)(\lambda x.\,ux)$, definitionally $M$. Their gradings are $1$ and $0$, so parity excludes every cell and cell sequence. A hypothetical MLTT-native equality would likewise give $\mathsf{False}=\mathsf{True}$; this is \code{commonSourceNoNativePathEquality}. The declarations \code{commonSourceNoStep2} and \code{commonSourceNoPath2} prove the separation in our two-dimensional theory. Moreover,
\[
L_{\beta\eta}=P_\beta\concatop P_\eta^{-1}:\Path(M,M)
\]
is a canonical loop in our two-dimensional theory with parity $1$; \code{commonSourceBetaEtaLoopNontrivial} proves that no $\PathTwo$ connects it to $\Nil$. This is non-triviality in our two-dimensional theory, not yet a construction of a higher-inductive $S^1$ or a proof that its fundamental group is $\mathbb Z$.
Thus, in the present two-dimensional presentation, the combination
\(P_\beta\concatop P_\eta^{-1}\) produces a non-trivial syntactic loop. Its
non-triviality is established constructively by the parity invariant and does
not invoke univalence.
\end{proof}

\begin{example}[Comparing selected MLTT-native and two-dimensional witnesses]\label{ex:native}
For the common source $M$ of Corollary~\ref{cor:contra2}, both routes justify $M\defineq ua$ in the judgemental beta/eta presentation of Section~\ref{sec:judgemental}. Representing these conversions by MLTT-native reflexivity yields
\[
e_\eta:=\mathsf{refl}^{\mathrm{Id}}_M:\mathrm{Id}_B(M,ua),\qquad
e_\beta:=\mathsf{refl}^{\mathrm{Id}}_M:\mathrm{Id}_B(M,ua).
\]
Both types convert to $ua=ua$ and both witnesses are definitionally the same in MLTT. In Idris, \code{nativeEtaEvidence} and \code{nativeBetaEvidence} are the selected MLTT-native images of the routes after endpoint conversion; \code{nativeBetaEtaSameEvidence} proves their equality. Hence, for every observation $O:(ua=ua)\to\{0,1\}$,
\[
O(e_\eta)=O(e_\beta),\qquad
\parity(P_\eta)=0\ne1=\parity(P_\beta).
\]
The first equality is \code{nativeEvidenceObservationsAgree}; the parity comparison survives all 2-cells of our two-dimensional theory. Thus $P_\eta$ and $P_\beta$ are unequal \emph{Path} data in our two-dimensional theory, while their selected MLTT-native images coincide: the translation forgets the beta/eta labels. This is the intended collapse under MLTT-native reflexivity. The example asserts neither uniqueness of arbitrary identity proofs nor a general translation of $\Step$ into MLTT-native identity.
\end{example}

\begin{remark}[Intensionality with respect to MLTT]\label{rem:intensional}
	Corollary~\ref{cor:contra2} together with Example~\ref{ex:native} shows that our typed system of order-2 $\lambda$-calculus is \emph{really intensional} with respect to the native MLTT core of Idris: the computation rules of Idris identify $\Eta$ and $\Beta$ witnesses by reduction, whereas in $\Path$ the corresponding paths are distinct, and the parity invariant certifies that no pasting of 2-cells can ever identify them. The system therefore does not validate the identification of $\beta$- and $\eta$-contractions as equal evidence. The philosophical significance of this separation --- evidence with content, intensionality as a theorem, and the finer-grained computational plane --- is developed in Section~\ref{sec:philosophical}.
\end{remark}

\section{Equality evidence and the identity of proofs}\label{sec:philosophical}

The formal results of the previous sections admit a coherent philosophical interpretation, spelled out here for the readership of a logic journal: the theory is a \emph{proposal about what equality evidence is} and about what constructivism should demand of it. 

\noindent\textbf{Constructivism and the BHK reading of evidence.} According to the Brouwer--Heyting--Kolmogorov (BHK) interpretation, a proof is a construction that transforms its data; in particular, a proof of an equality between terms should be a \emph{procedure} transforming one term into the other. Intensional MLTT honours this reading only at the base: $\mathsf{refl}$ is canonical, but the $J$-eliminator is a postulate without computational content, and an arbitrary witness of $\mathsf{Id}_A(a,b)$ records no computation at all. Our $\Step$ and $\Path$ types are a BHK reading made literal: each rule states \emph{what counts as evidence} --- $\Beta$ and $\Eta$ are rewriting procedures, $\ApCong$ and $\LamCong$ procedures applied under context, $\refl$, $\sym$, $\trans$ the structural operations --- and recursion over $\Path$ is induction over constructions, with any step as a base case. Equality evidence is thus required to \emph{be} a computation, in the spirit of Brouwer's ``no mathematical truth without a construction'', rather than certified by reflexivity and transported along definitional coercion.

\noindent\textbf{Intensionality: the identity of evidence.}
The intensional/extensional distinction is the Fregean question of whether identity of \emph{sense} is recorded when identity of \emph{reference} is asserted; in type theory the question is proof-relevant: do identity proofs carry information beyond the equality of the endpoints? MLTT answers ``yes'' syntactically, but its definitional equality then \emph{collapses senses}: in the native core of Idris, the $\beta$-redex $(\lambda t.\, t\, a)(\lambda x.\, u\, x)$ and the $\eta$-redex $(\lambda t.\, t)((\lambda x.\, u\, x)\, a)$ are one and the same term $u\, a$, so no proposition of the core can even \emph{express} their difference. Section~\ref{sec:consistency} upgrades this weak, accidental intensionality into a theorem: the parity invariant proves that the sense ``apply $\beta$'' and the sense ``apply $\eta$'' are \emph{provably distinct presentations} of the same reference $u\, a$ (Corollary~\ref{cor:contra2}). Intensionality is thus a positive, witnessed property: not merely the absence of a proof of coincidence, but a proof of non-coincidence --- the theory distinguishes proofs of the same proposition by their computational content, which is precisely the Fregean demand made precise.

\noindent\textbf{Relation to homotopy type theory.}
Homotopy type theory organizes identity evidence into an $\infty$-groupoid and postulates univalence --- identity of types is equivalence of types. Two remarks place our theory relative to this program. (i) Univalence, like $\mathsf{StepNat}$, is a postulate without computational content; our order-2 layer shows that the \emph{low-dimensional} groupoidal structure need not be postulated: the typed $2\beta/2\eta$ squares of Section~\ref{sec:step2} and the coherence laws of Section~\ref{sec:algebra} derive the pasting structure from reduction rules. (ii) HoTT is proof-relevant, yet it inherits MLTT's definitional collapse of $\beta$- and $\eta$-redexes; our parity invariant shows that keeping them apart is \emph{compatible} with the full 2-dimensional coherence --- the groupoid laws do not force the identification of distinct rewrites. The theory can thus be read as a constructive, proof-relevant refinement of the identity machinery shared by MLTT and HoTT, in which higher evidence is generated by computation rather than postulated by induction principles. 

HoTT already derives groupoid operations, transport and naturality of identity-valued homotopies by path induction
\cite[Secs.~2.1--2.4, especially Lemma~2.4.3]{hottbook}. In its standard
presentation, the circle is a higher-inductive type with a base point and a
generating loop, while the usual encode--decode calculation
\(\pi_1(S^1)\simeq\mathbb Z\) uses a universe-valued cover built with
univalence \cite[Secs.~6.4 and 8.1]{hottbook}. The present result instead
concerns a specified algebra of labelled conversions and generated 2-cells:
\(L_{\beta\eta}\) is a labelled loop in our two-dimensional theory with a parity separation, not yet
that higher-inductive circle or an integer classification. However, 
if a circle in our two-dimensional theory—along with its loop group structure and an integer classification—is
subsequently constructed directly from the tagged syntax, that internal proof
could be carried out without using the univalence axiom (work in progress). This would represent an
advantage over HoTT, since incorporating the univalence axiom into MLTT complicates the proof of confluence in HoTT (which remains unproven) and thereby hinders the decidability of the identity.

\noindent\textbf{Relation to $LND_{EQ}-TRS$ \cite{ramos2026calculus}.} 
It is also an intentional theory governed by the same rules as Step. Although it shares the structural rules of \code{Step2}, it does not include the typed $2\beta$ and $2\eta$ squares with their coherences.  Furthermore in  the explicit computational-path setting of Ramos et al., homotopies are likewise presented pointwise, but each component is a syntactic or computational path; the
corresponding \code{NatPath} (or \code{NatStep}, in the terminology used here)
is a directed rewriting rule over those pointwise paths
\cite{ramos2018explicit}. Our theory makes a different choice:
\code{HomotopyStep} is a semantic interface, and \code{HomotopyPath} stores a
global sequence whose entries are such families, whereas \code{HomotopyStepI}
and \code{HomotopyPathI} inductively record a global sequence of presented
steps between functions. Therefore the pointwise \code{NatPath} and
\code{StepNat} claims do not follow in the theory; the derivable
naturality theorems are the ones built from these global inductive homotopies
and their certified evaluators, and proved using our higher-level steps
\code{Beta2Step}, \code{Eta2Step}, \code{apCongStep2}, 
\code{LamCong2Step} etc.

There is also a methodological difference in the algebraic layer. In the
present theory, the algebraic properties of paths are proved by structural recursion over sequences of one-step conversions; in $LND_{EQ}-TRS$, the corresponding properties are given
directly as rules over computational paths.

\section{Conclusions and future work}\label{sec:conclusion}
We have presented a typed order-2 $\lambda$-calculus, based on the higher $\lambda$-models of \cite{martinez2023arbitrary}, whose equality evidence is computational: paths are explicit sequences of one-step conversions, 2-cells are explicit pastings of typed $2\beta$ and $2\eta$ squares, and every coherence, naturality and transport law in the computational plane is proved by recursion, with \emph{any} step available as a base case --- no $J$-eliminator, no postulates. Three findings summarize the paper.

\begin{enumerate}[(1)]
	\item \textbf{Computable naturality.} Naturality is a theorem for inductive homotopies with evaluations ($\NatStepI$, $\NatPathI$), while for \emph{semantic} homotopies do not hold in our theory.
	\item \textbf{Consistency and intensionality.} The parity invariant proves consistency and the central negative result: the $\beta$- and $\eta$-contractions are distinct evidence and no 2-cell can identify them, whereas the native MLTT core of Idris identifies them definitionally. Intensionality is here a \emph{theorem}, not an accident.
	\item \textbf{A constructive philosophy.} Read philosophically (Section~\ref{sec:philosophical}), the theory is a BHK-style constructivism in which evidence \emph{is} computation, a proof-relevant intensionality that keeps senses distinct, and an explicit boundary --- drawn inside the formal system --- between effective procedures and arbitrary functions.
\end{enumerate}
\emph{Future work.}  (i) Construct the type $S^1$ and calculate its fundamental group using the canonical non-trivial loop $L_{\beta\eta}$, without using univalence. (ii) Complete the coherence laws of $\PathTwo$ (an Eckmann--Hilton argument);  (iii) promote the parity invariant to a full logical relation, for normalization and canonicity; (iv) extend the system to order $\infty$ following \cite{martinez2023groupoid,lurie2009}.

\backmatter
\section*{Declarations}
\textbf{Code availability.} The accompanying sources are \code{Path.idr},
\code{PathExamples.idr}, \code{PathNonComputable.idr},
\code{PathStepNatRefutation.idr}, the \code{StepNatProof} modules, and
\code{Path.ipkg}.
Their roles and verification command are given in Appendix~\ref{sec:appendix}.

\textbf{Author information.} The three author affiliations are listed on the
title page.

\textbf{Funding.} No external funding was received.

\textbf{Competing interests.} The authors declare that they have no competing
interests.

\begin{appendices}
\section{Correspondence with the formalization}\label{sec:appendix}

The accompanying development was checked with Idris 2 version 0.8.0
(commit \code{acde1c926}). The three modules use \code{\%default total}
\cite[Theorem Proving: Totality Checking]{idris2}. In particular, a statement of type
$T\to\mathsf{Void}$ is checked as a total elimination of every possible input,
not accepted merely because one impossible constructor case was written.
A parallel Lean formalization is published in the Palomar registry~\cite{palomar2026lean}.
The build command, from the project directory, is
\begin{verbatim}
idris2 --build Path.ipkg
\end{verbatim}

\begin{center}
\scriptsize
\setlength{\tabcolsep}{2pt}
\renewcommand{\arraystretch}{0.78}
\begin{tabular}{@{}p{0.44\textwidth}p{0.51\textwidth}@{}}
\toprule
Mathematical statement & Idris declarations\\
\midrule
Steps, paths, and 2-cells (Defs.~\ref{def:step}, \ref{def:path}, \ref{def:step2}) &
\code{Step}, \code{Path}, \code{Step2}\\
Path algebra (Sec.~\ref{sec:algebra}) &
\code{leftInv2}, \code{rightInv2}, \code{assoc2}, \code{invDistrib}, \code{invInv}\\
Beta, eta and lambda-congruence squares &
  \code{Beta2Step}, \code{Eta2Step}, \code{LamCong2Step}; path-level builders \code{beta2}, \code{eta2}, \code{lamCong2}\\
Certified step evaluation (Thm.~\ref{prop:evalhsi}) &
\code{evalHomotopyStepIHom}, \code{EvalHSI}, \code{EHSRefl}, \code{EHSBeta}, \code{EHSEta}, \code{EHSApCong}, \code{EHSLamCong}, \code{EHSSym}, \code{EHSTrans}, \code{evalHomotopyStepI}\\
Function-space evaluation (Def.~\ref{def:evalhpi}) &
\code{EvalHPI}, \code{evalHomotopyPathI}\\
Direct naturality (Thms.~\ref{thm:nattransi}, \ref{thm:nattranshp2}) &
\code{NatStepI}, \code{NatPathI}\\
Conjugation and transport (Sec.~\ref{sec:path2}) &
\code{ConjApcongStep2Eval}, \code{Transp2viaHP}, \code{Transp2viaHPResp}, \code{Transp2viaHPApcongEval}, \code{Untransp2viaHP}\\
Parity and non-collapse (Sec.~\ref{sec:consistency}) &
\code{step2PreservesParity}, \code{path2PreservesParity}, \code{contraexample1}, \code{contraexample1Nil}, \code{contraexample2}, \code{contraexample2Path2}\\
Common-source separation (Cor.~\ref{cor:contra2}) &
\code{commonSourceBetaStep}, \code{commonSourceEtaStep}, \code{commonSourceBetaPath}, \code{commonSourceEtaPath}, \code{commonSourceNoStep2}, \code{commonSourceNoPath2}, \code{commonSourceNoNativePathEquality}, \code{commonSourceDistinctWitnesses}, \code{commonSourceBetaEtaLoop}, \code{commonSourceBetaEtaLoopNontrivial}\\
MLTT-native witness comparison (Ex.~\ref{ex:native}) &
\code{nativeBetaEtaSameTerm}, \code{nativeBetaEtaSameEvidence}, \code{nativeEvidenceObservationsAgree}\\
Constructive refutation of \code{StepNat} &
\code{StepNat}, \code{notStepNat}, \code{notNatPath}, \code{decStepNat}\\
\bottomrule
\end{tabular}
\end{center}

\code{Path.idr} contains the presentation, path algebra, certificate builders,
and transport constructions. \code{PathExamples.idr} proves the invariant and
separation results. \code{PathNonComputable.idr} declares the aliases
\code{StepNat} and \code{NatPath} without inhabitants; their refutations are exported by
\code{PathStepNatRefutation.idr}. There is no assumed inhabitant. Type checking
certifies the displayed derivations relative to
the implementation and trusted MLTT/Idris host; it does not certify the cited papers,
a semantic interpretation, or the philosophical analysis.
\end{appendices}

\begingroup
\interlinepenalty=10000
\setlength{\bibsep}{0.1em}
\renewcommand{\bibfont}{\fontfamily{\rmdefault}\fontsize{7pt}{8pt}\selectfont}
\bibliography{path-article}
\endgroup
\end{document}